\newif\ifappendix\appendixfalse
\makeatletter \@input{flags} \makeatother

\documentclass{article}

\usepackage{mathpartir}
\usepackage{ifluatex}
\ifluatex{}
\usepackage{unicode-math}
\usepackage{fontspec}
\fi{}
\usepackage{amsmath}
\usepackage{amsthm}
\ifluatex{}
\else{}
\usepackage{amssymb}
\fi{}
\usepackage{stmaryrd}
\usepackage[hidelinks=true]{hyperref}
\usepackage{cleveref}
\usepackage{tikz-cd}
\usepackage[backend=biber, style=alphabetic]{biblatex}
\usepackage[margin=false,inline=true]{fixme}

\theoremstyle{plain}
\newtheorem{theorem}{Theorem}[section]
\newtheorem{lemma}[theorem]{Lemma}

\newtheorem{corollary}[theorem]{Corollary}

\theoremstyle{definition}
\newtheorem{definition}[theorem]{Definition}
\newtheorem{example}[theorem]{Example}

\theoremstyle{remark}
\newtheorem{remark}[theorem]{Remark}

\ifluatex{}
\else
\usepackage[tt=false]{libertine}
\usepackage[notextcomp]{stix2}
\fi{}

\bibliography{refs.bib}

\FXRegisterAuthor{aaa}{anaaa}{\color{cyan}AAA}

\DeclareMathOperator{\CPO}{\mathsf{CPO}}
\DeclareMathOperator{\CLat}{\mathsf{CLat}_\land}
\DeclareMathOperator{\N}{\mathbb{N}}
\DeclareMathOperator{\Mono}{\xrightarrow{\mathrm{mono}}}
\DeclareMathOperator{\Cont}{\xrightarrow{\mathrm{cont}}}

\DeclareMathOperator{\im}{\mathrm{im}}

\newcommand{\later}{\smalltriangleright}

\title{On Pitts' Relational Properties of Domains}
\author{Arthur Azevedo de Amorim}

\begin{document}

\maketitle

\begin{abstract}
  Andrew Pitts' framework of \emph{relational properties of domains} is a
  powerful method for defining predicates or relations on domains, with
  applications ranging from reasoning principles for program equivalence to
  proofs of adequacy connecting denotational and operational semantics.  Its
  main appeal is handling recursive definitions that are not obviously
  well-founded: as long as the corresponding domain is also defined recursively,
  and its recursion pattern lines up appropriately with the definition of the
  relations, the framework can guarantee their existence.

  Pitts' original development used the Knaster-Tarski fixed-point theorem as a
  key ingredient.  In these notes, I show how his construction can be seen as an
  instance of other key fixed-point theorems: the \emph{inverse limit
    construction}, the \emph{Banach fixed-point theorem} and the \emph{Kleene
    fixed-point theorem}.  The connection underscores how Pitts' construction is
  intimately tied to the methods for constructing the base recursive domains
  themselves, and also to techniques based on \emph{guarded recursion}, or
  \emph{step-indexing}, that have become popular in the last two decades.
\end{abstract}

\section{The Original Result}
\label{sec:original}

When reasoning about programs, it is common to compare their behaviors.  We
might ask if two programs behave equivalently, if their public outputs are
equal, or if one program terminates more often than the other, among other
questions.  Many of these issues can be phrased naturally using recursive
relations.  For example, to argue that two functions are equivalent, we might
want to check if they produce equivalent outputs when applied to equivalent
inputs, for some notion of equivalence.  However, when reasoning about
higher-order or stateful programs, equivalence for inputs and outputs is defined
in terms of equivalence for arbitrary programs.  Thus, we end up with a circular
definition of equivalence, which requires care to justify formally without
running into paradoxes.

Andrew Pitts' framework of \emph{relational properties of
  domains}~\cite{Pitts96} is a powerful tool for constructing such relations.
We can summarize the idea as follows.

\begin{theorem}[\cite{Pitts96}]
  \label{thm:minimal-invariant-relations}
  Let $D$ be an object of a pointed $\CPO$-category $\mathcal{C}$.  Suppose that
  $D$ is equipped with an isomorphism $i : F(D,D) \cong D$ that satisfies the
  minimal invariant property, where
  $F : \mathcal{C}^{op} \times \mathcal{C} \to \mathcal{C}$ is a $\CPO$-functor.
  Suppose moreover that $\mathcal{C}$ is equipped with an admissible relational
  structure $\mathcal{R}$, and that $F$ acts on $\mathcal{R}$.  Then there
  exists $R_D \in \mathcal{R}_D$ such that $R_D = (i^{-1})^*F(R_D,R_D)$.
\end{theorem}

Here is how we can read this result intuitively, before diving into formal
definitions. The object $D$ is a universe where we model the behavior of the
programs.  In Pitts' original result, $D$ was assumed to be a \emph{complete
  partial order}, or CPO, a domain-theoretic notion for modeling general
recursion and nontermination.  Here, instead, we assume that $D$ lives in some
pointed $\CPO$-category $\mathcal{C}$, a generalization allows us to carry the
core of Pitts' arguments while accounting for variations that have been explored
in the literature, such as families of CPOs~\cite{AmorimFJ20}, diagrams of
CPOs~\cite{Levy02}, or CPOs equipped with a metric~\cite{AmorimGHKC17}.

We assume that $D$ is defined recursively as $F(D,D) \cong D$.  The equation is
stated using a functor $F$, where each recursive occurrence of $D$ is either
contravariant or covariant; being a $\CPO$-functor simply means that $F$
interacts well with the structure of $\mathcal{C}$.  In principle, there could
be many solutions to such equations, but \Cref{thm:minimal-invariant-relations}
only applies to those that satisfy the \emph{minimal invariant property}, which
roughly means that $D$ is completely characterized by repeatedly unfolding its
definition.

The conclusion of the theorem says that we can construct some ``relation'' $R_D$
on $D$.  In the applications we sketched above, $R_D$ could be a binary relation
on a CPO, but the result applies to other settings as well, such as relations of
different arities or families of relations.  The relational structure
$\mathcal{R}$ formalizes which properties are required of the notion of
``relation'' for the construction to apply.  The definition of $R_D$ is given by
a recursive equation $R_D = (i^{-1})^*F(R_D,R_D)$, which is derived from an
action of $F$ on $\mathcal{R}$.  Different actions and relational structures
yield different definitions, and it is our job to choose them appropriately
depending on the application at hand.

\medskip

Let us now spell out how this works in detail. A \emph{complete partial order}
(CPO) is a poset $(X,\sqsubseteq)$ such that every increasing chain
$x : \N \Mono X$ has a \emph{limit}, or least upper bound, denoted
$\lim_n x(n)$.  A CPO $X$ is \emph{pointed} if it has a least element
$\bot \in X$. A function $f : X \to Y$ between CPOs is \emph{continuous},
denoted $f : X \Cont Y$, if it is monotone and preserves limits.  CPOs and
continuous functions between them form a category $\CPO$.  This category is
cartesian closed; the exponential $Y^X$ is given by the set of continuous
functions of type $X \Cont Y$ ordered pointwise.

A $\CPO$-category is a category $\mathcal{C}$ where the sets of morphisms
$\mathcal{C}(X,Y)$ are CPOs, and such that composition
$(-) \circ (-) : \mathcal{C}(Y,Z) \times \mathcal{C}(X,Y) \to \mathcal{C}(X,Z)$
is continuous.  (Product CPOs are ordered component-wise.)  The most basic
example of $\CPO$-category is $\CPO$ itself, for the order relation on
continuous functions defined above.  Another example is given by \emph{functor
  categories} of the form $\mathcal{C}^I$, where $\mathcal{C}$ is a
$\CPO$-category and $I$ is a small category.  A morphism in $\mathcal{C}^I$ is a
family of arrows $(X_i \to Y_i)_{i \in I}$, and we obtain a $\CPO$-category by
ordering such families pointwise.  Combined with the previous example, this
shows that families or diagrams of CPOs also form $\CPO$-categories.  If
$\mathcal{C}$ is a $\CPO$-category, then so is $\mathcal{C}^{op}$, by inheriting
the structure on $\mathcal{C}$. A $\CPO$-functor
$F : \mathcal{C} \to \mathcal{D}$ is a functor whose action on morphisms is a
continuous function.

Given a $\CPO$-category, we say that $Y \in \mathcal{C}$ is \emph{pointed} if
$\mathcal{C}(X,Y)$ is pointed for every $X$, and if $\bot \circ f = \bot$ for
every $f$.  When $\mathcal{C} = \CPO$, this definition of pointedness coincides
with the one given above.  Every terminal object 1 is pointed: if $X$ is an
object, the unique arrow of type $X \to 1$ is the least element.  We say that
$\mathcal{C}$ itself is pointed if every object is pointed and it has a terminal
object 1.  In this case, any $\bot : X \to Y$ in $\mathcal{C}$ factors through
1.  For example, $\CPO$ is not a pointed $\CPO$-category (because not every CPO
is pointed according to our definition), but we do obtain a pointed
$\CPO$-category by restricting ourselves to pointed CPOs.  (More generally, any
$\CPO$-category $\mathcal{C}$ with a terminal object has a pointed counterpart
$\mathcal{C}_\bot$ obtained by restricting $\mathcal{C}$ to pointed objects.)

Given an isomorphism $i : F(D,D) \cong D$ in $\mathcal{C}$, where $D$ is
pointed, we say that $D$ has the \emph{minimal invariant property} if the
following condition holds.  First, given $\pi : D \Cont D$, we define
$\Phi(\pi) : D \Cont D$ as
$\Phi(\pi) \triangleq i \circ F(\pi,\pi) \circ i^{-1}$, and pose
$\pi_i \triangleq \Phi^i(\bot)$.  Intuitively, each $\pi_i$ is a projection
function that truncates $D$ to allow for at most $i$ unfoldings of its
definition; hence, $\lim_i \pi_i$ is a projection function that allows for an
arbitrary number of unfoldings.  Intuitively, we would expect $\lim_i \pi_i$ to
be the identity on $D$, because leaving the number of unfoldings unbounded
should be tantamount to not truncating $D$ at all.  However, this is not
necessarily true; the best we can show in general is
$\lim_i \pi_i\sqsubseteq 1_D$.  The minimal invariant property says precisely
that we can strengthen this inequality to $\lim_i \pi_i = 1_D$.

Given a category $\mathcal{C}$, a \emph{relational structure} on $\mathcal{C}$
is simply a functor $\mathcal{R} : \mathcal{C}^{op} \to \CLat$, where $\CLat$ is
the category of complete lattices and functions that preserve greatest lower
bounds.  We denote the value of $\mathcal{R}$ at some object $X \in \mathcal{C}$
as $\mathcal{R}_X$, and we use the variables $R$, $S$ and $T$ to range over the
elements of $\mathcal{R}_X$.  If $f : X \to Y$ is a morphism in $\mathcal{C}$,
we write $\mathcal{R}(f)$ as $f^*$ when $\mathcal{R}$ can be understood from the
context.

\begin{example}
  \label{ex:relational-structure}
  Our motivating example of relational structure is the one obtained by choosing
  $\mathcal{C} = \CPO$, and posing $\mathcal{R}_X$ to be the set of binary
  relations between the elements of $X$ ordered by inclusion.  The greatest lower
  bound of a family of relations is simply their \emph{intersection}.  And the
  action of a continuous function $f : X \Cont Y$ on $\mathcal{R}_Y$ takes the
  \emph{inverse image} of a relation by $f$.  For intuition, we'll keep this
  vocabulary when discussing other relational structures as well.
\end{example}

To define \emph{admissible} relational structures, it is convenient to shift our
perspective a bit.  Given a relational structure $\mathcal{R}$ on $\mathcal{C}$,
we can build a category, also denoted $\mathcal{R}$, as follows.  The objects of
$\mathcal{R}$ are pairs $(X,R)$, where $X \in \mathcal{C}$ and
$R \in \mathcal{R}_X$.  (By abuse of notation, I'll often use $R$ to represent
the object $(X,R) \in \mathcal{R}$.)  A morphism $f : (X,R) \to (Y,S)$ is a
morphism $f : X \to Y$ in $\mathcal{C}$ such that $R \leq f^*S$.  In terms of
\Cref{ex:relational-structure}, this simply means that the function $f$ takes
elements related by $R$ to elements related by $S$.  We can check that
identities and composition in $\mathcal{C}$ can be lifted to the morphisms of
$\mathcal{R}$.  We have a canonical functor $p : \mathcal{R} \to \mathcal{C}$
that maps the object $(X,R)$ to $X$ and acts as the identity on morphisms.

By unfolding definitions, we can restate some of the properties of the functor
$\mathcal{R}$ in terms of the above construction:
\begin{lemma}
  \label{lem:relational-structure-properties}
  Let $X$, $Y$ and $Z$ be arbitrary objects of $\mathcal{C}$.
  \begin{enumerate}
  \item $1_X : R \to S \iff R \leq S$, for all $R, S \in \mathcal{R}_X$.
  \item $f : f^*S \to S$ for any $S \in \mathcal{R}_Y$ and $f : X \to Y$.
  \item $gf : R \to S \iff f : R \to g^*S$, for all $f : X \to Y$,
    $g : Y \to Z$, $S \in \mathcal{R}_Z$ and $R \in \mathcal{R}_X$.
  \item $f : R \to \bigcap_{i \in I} S_i \iff \forall i \in I, f : R \to S_i$,
    for any index set $I$, $f : X \to Y$, $R \in \mathcal{R}_X$ and
    $S \in \mathcal{R}_Y^I$.
  \end{enumerate}
\end{lemma}

If $\mathcal{R}$ is a relational structure over a \emph{pointed} $\CPO$-category
$\mathcal{C}$, we say that a relation $S \in \mathcal{R}_Y$ is \emph{admissible}
if the following conditions hold for all $R \in \mathcal{R}_X$.  First,
$\bot : R \to S$; second, $\lim_i f_i : R \to S$ whenever
$(f_i : R \to S)_{i \in \N}$ is an increasing sequence of morphisms.  We say
that $\mathcal{R}$ itself is admissible if every relation is admissible.
Intuitively, being admissible means that a relation always holds of diverging
programs and is compatible with recursive program definitions, which are
constructed using limits via Kleene's fixed point theorem.

\begin{example}
  \label{ex:admissible-relational-structure}
  We can adapt \Cref{ex:relational-structure} to obtain an admissible relational
  structure as follows.  First, instead of considering arbitrary CPOs, we just
  consider pointed ones; that is, we take $\mathcal{C} = \CPO_\bot$.  Second,
  instead of considering arbitrary relations, we consider only those that
  contain $\bot$ and are closed under taking limits of chains.
\end{example}

The missing piece in the statement of \Cref{thm:minimal-invariant-relations} is
what it means for a $\CPO$-functor
$F : \mathcal{C}^{op} \times \mathcal{C} \to \mathcal{C}$ to act on
$\mathcal{R}$.  For each $R \in \mathcal{R}_X$ and $S \in \mathcal{R}_Y$, we
assume that there is some $F(R,S) \in \mathcal{R}_{F(X,Y)}$; moreover, if
$f : R' \to R$ and $g : S \to S'$ are morphisms in $\mathcal{R}$, then $F(f,g)$
should be a morphism of type $F(R,S) \to F(R',S')$ in $\mathcal{R}$. (Note the
contravariance on first argument).

We can now sketch the main idea of Pitts' original construction.  For the rest
of the paper, we fix some pointed $\CPO$-category $\mathcal{C}$ equipped with an
admissible relational structure $\mathcal{R}$, a $\CPO$-functor
$F : \mathcal{C}^{op} \times \mathcal{C} \to \mathcal{C}$ with an action on
$\mathcal{R}$, and an object $D$ that satisfies the minimal invariant property
for an isomorphism $i : F(D,D) \cong D$.

\begin{proof}[Proof of \Cref{thm:minimal-invariant-relations}]
  The proof relies on the Knaster-Tarski fixed point theorem: every monotone
  function on a complete lattice has a least fixed point.  Since the mapping
  $R \mapsto (i^{-1})^*F(R,R)$ is not monotone, we need to modify its definition
  a bit.  Pitts' employed the trick of separating covariant and contravariant
  arguments: if we pose $L \triangleq \mathcal{R}_D^{op} \times \mathcal{R}_D$,
  then the function
  \begin{align*}
    \Psi & : L \to L \\
    \Psi(R^{-},R^{+}) & \triangleq ((i^{-1})^*F(R^+,R^-), (i^{-1})^*F(R^-,R^+))
  \end{align*}
  \emph{is} monotone, and we can construct a least fixed point
  $(R_D^{-}, R_D^{+})$.  Note that $(R_D^{+}, R_D^{-})$ is also a fixed point,
  so $(R_D^{-},R_D^{+}) \leq (R_D^+,R_D^-)$ in $L$, and thus
  $R_D^{+} \leq R_D^{-}$.  To finish the proof, we just need to show the reverse
  inequality.  This is where the properties of minimal invariant and the
  relational structures come into play.  We can show that
  $\pi_i : R_D^{-} \to R_D^{+}$ by induction on $i$, which implies, by
  admissibility, that $1 = \lim_i \pi_i : R_D^{-} \to R_D^{+}$.  But this is
  equivalent to $R_D^{-} \leq R_D^+$ by
  \Cref{lem:relational-structure-properties}, from which the result follows.
\end{proof}

\begin{remark}[Uniformity]
  \label{rem:uniformity}
  This proof shows that a stronger result holds: for all
  $i \in \N$,
  \begin{align}
    \label{eq:uniform}
    \pi_i & : R_D \to R_D.
  \end{align}
  Intuitively, this means that the constructed relation $R_D$ still holds after
  we truncate an element of $D$ after $i$ unfoldings.  This property, known as
  \emph{uniformity}, will play an important role in \Cref{sec:banach}, when
  constructing $R_D$ by the Banach fixed-point theorem.
\end{remark}

\begin{remark}
  Pitts' presentation differs from mine in a few respects~\cite{Pitts96}.  What
  I call a relational structure here corresponds to what he calls a relational
  structure with inverse images and intersections.  More importantly, his notion
  of action on a relational structure is different: rather than requiring
  $\mathcal{R}$ to be admissible, he requires $F(R,S)$ to be admissible whenever
  $S$ is.  This is a strengthening of the above notion of action, since it must
  be defined even for relations that are not admissible.  It allows us to
  formulate more useful coinduction principles associated with the relation
  $R_D$, but it does not change the construction of $R_D$ itself, which is why
  we do not consider it here.
\end{remark}

\section{Inverse Limit Construction}
\label{sec:inverse-limit}

In practice, minimal invariants such as $D$ are often obtained with Scott's
\emph{inverse limit construction}.  The method can be seen as an adaptation of
Kleene's fixed-point theorem that accounts for mixed-variance functors, and can
be carried out for many $\CPO$-categories~\cite{Wand79,SmythP82}.  After
reviewing the idea, we will see that Pitts' result,
\Cref{thm:minimal-invariant-relations}, is just an instance of it!

We say that two morphisms $f^e : X \to Y$ and $f^p : Y \to X$ in $\mathcal{C}$
form an \emph{embedding-projection pair} if $f^pf^e = 1_X$ and
$f^ef^p \sqsubseteq 1_Y$.  We can show that each half of the pair uniquely
determines the other.  Embeddings and projections compose, so we can form a
subcategory $\mathcal{C}^e$ consisting of all embeddings, and $\mathcal{C}^p$
consisting of all projections.  With embeddings and projections, we can make
mixed-variance functors more symmetric. Since $F$ is a $\CPO$-functor, its
action on morphisms is monotone, and we can show that $F(f^p,f^e)$ is an
embedding, with $F(f^e, f^p)$ being the corresponding projection.  Thus, $F$
determines a functor $F^e : \mathcal{C}^e \to \mathcal{C}^e$ by posing
$F^e(X) \triangleq F(X,X)$ on objects, and $F^e(f^e) \triangleq F^e(f^p, f^e)$
on morphisms.

Much like Kleene's fixed-point theorem, we'll see that we can build $D$ by
considering a chain of finite iterations of $F^e$ and taking its
colimit---which, in the context of Kleene's construction, would just correspond
to a limit in a CPO.  Since we are dealing with embeddings, colimits behave
particularly symmetrically, a phenomenon known in the literature as the
\emph{limit-colimit coincidence}:

\begin{theorem}[\cite{SmythP82}]
  \label{thm:bilimit-definition}
  Let $X^e : (\N,\leq) \to \mathcal{C}^e$ be a diagram of embeddings, which
  uniquely corresponds to a diagram $X^p : (\N,\geq) \to \mathcal{C}^p$ of
  projections.  Let $A \in \mathcal{C}$. The following conditions are
  equivalent.
  \begin{itemize}
  \item $A$ is a colimit of $X^e$ in $\mathcal{C}$.
  \item $A$ is a limit of $X^p$ in $\mathcal{C}$.
  \item There is a cocone of embeddings $f^e : X^e \to \Delta A$ such that
    $\lim_i f_i^e \circ f_i^p = 1_A$.
  \item There is a cone of projections $f^p : \Delta A \to X^p$ such that
    $\lim_i f_i^e \circ f_i^p = 1_A$.
  \end{itemize}
  In this situation, $f^e : X^e \to \Delta A$ is a colimiting cocone, and
  $f^p : \Delta A \to X^p$ is a limiting cone.  We call the pair $(A,f)$ the
  \emph{bilimit} of $X$.
\end{theorem}

Because of this result, we can show that $F^e$ preserves bilimits of chains of
embeddings in $\mathcal{C}$.  Then, constructing the fixed point of $F$ becomes
simply a matter of adapting the proof of Kleene's fixed-point theorem.

\begin{theorem}[\cite{SmythP82}]
  \label{thm:inverse-limit-construction}
  Suppose that $\mathcal{C}$ has bilimits of chains of embeddings. Then $F$ has
  a minimal invariant $i : F(D,D) \cong D$.
\end{theorem}

\begin{proof}
  Let $X_i = (F^e)^i(1)$.  Since $\mathcal{C}$ is pointed, there is an embedding
  $f_0^e = \bot : 1 \to X_1$.  By iterating $F^e$ on $f_0^e$, we can construct a
  sequence of embeddings $X_i \to X_{i+1}$.  By hypothesis, this chain has a
  bilimit, which we call $g^e : X \to \Delta D$.  Since $F^e$ preserves bilimits
  of embeddings, we know that
  $F^e(g) : F^e(X) = (X_i)_{i \geq 1} \to \Delta F^e(D)$ is a bilimit.  Note
  that $1$ is an initial object of $\mathcal{C}^e$, so we can extend this cocone
  to $h^e : X \to \Delta F^e(D)$ by posing
  \begin{align*}
    h_0^e & : 1 \to F^e(D) \\
    h_0^e & \triangleq \bot \\
    h_{i+1}^e & : F^e(X_i) \to F^e(D) \\
    h_{i+1}^e & \triangleq F^e(g_i).
  \end{align*}
  Since both $F(D,D)$ and $D$ satisfy the same universal property, we get an
  isomorphism $i : F(D,D) \cong D$.  The construction of this isomorphism
  implies, for every $j \in \N$,
  \begin{align*}
    i \circ h_{j+1}^e & = i \circ F(g_j^p, g_j^e) = g_{j+1}^e.
  \end{align*}
  Taking projections on both sides, we obtain
  \begin{align*}
    F(g_j^e,g_j^p) \circ i^{-1} & = g_{j+1}^p.
  \end{align*}
  Combining the two equations, we find
  \begin{align*}
    g_{j+1}^e \circ g_{j+1}^p & = i \circ F(g_j^pg_j^e,g_j^eg_j^p) \circ i^{-1}.
  \end{align*}
  Since $g_0^e = \bot$ and $g_0^p = \bot$, this implies that
  $\pi_j \triangleq g_j^e \circ g_j^p$ satisfies exactly the same equations as
  the projection functions used in the definition of the minimal invariant
  property.  By \Cref{thm:bilimit-definition}, the limit of this sequence is
  the identity on $D$, so $i : F(D,D) \cong D$ indeed satisfies the minimal
  invariant property.
\end{proof}

To see how this relates to Pitts' construction, note that $\mathcal{R}$ can also
be seen as a pointed $\CPO$-category, and the projection
$p : \mathcal{R} \to \mathcal{C}$ preserves this structure.  Indeed,
admissibility means that the morphisms of $\mathcal{R}$ have the structure of a
pointed CPO inherited from the morphisms of $\mathcal{C}$.  The terminal object
of $\mathcal{R}$ is just the terminal object of $\mathcal{C}$ equipped with the
greatest relation on $\mathcal{R}_1$, which exists because $\mathcal{R}_1$ is a
complete lattice.  Moreover, if we see each relation $F(R,S)$ as an object of
$\mathcal{R}$, then the action of $F$ on $\mathcal{R}$ can be described
equivalently as a $\CPO$-functor of type
$\mathcal{R}^{op} \times \mathcal{R} \to \mathcal{R}$ making the following
diagram commute:
\begin{center}
  \begin{tikzcd}
    \mathcal{R}^{op} \times \mathcal{R} \ar{r} \ar{d}
    & \mathcal{R} \ar{d} \\
    \mathcal{C}^{op} \times \mathcal{C} \ar{r}{F}
    & \mathcal{C}.
  \end{tikzcd}
\end{center}
To apply the inverse limit construction to this lifted functor, we just need to
show that $\mathcal{R}$ has bilimits of chains of embeddings.

\begin{lemma}
  \label{lem:bilimits-relational-structures}
  Admissible relational structures create bilimits.  That is, if $X$ is a chain
  of embeddings in $\mathcal{R}$, and $f : pX \to \Delta L$ is a bilimit in
  $\mathcal{C}$, then $f : X \to \Delta R$ is a bilimit in $\mathcal{R}$, where
  $R \in \mathcal{R}_L$ is defined as
  \begin{align*}
    R & \triangleq \bigcap_n (f_n^p)^*X_n.
  \end{align*}
\end{lemma}

\begin{proof}
  By \Cref{thm:bilimit-definition}, it suffices to show that the corresponding
  cone of projections is a limiting cone in $\mathcal{R}$.  We'll show that
  there is a bijective correspondence between morphisms of type $T \to R$ and
  cones of type $\Delta T \to X$ in $\mathcal{R}$ that is natural in $T$.

  In one direction, suppose that $g : T \to \bigcap_n (f_n^p)^*X_n$ is a
  morphism in $\mathcal{R}$.  This means that $f_n^pg : T \to X_n$ is a morphism
  for every $n \in \N$, and we can check that they form a cone $\Delta T \to X$.
  Conversely, suppose that we are given a cone $g : \Delta T \to X$.  By
  projecting this cone onto $\mathcal{C}$, we obtain another cone
  $pg : \Delta pT \to pX$.  Since $f^p : \Delta L \to pX$ is limiting, there is
  a unique mediating morphism $g' : pT \to L$.  Moreover, for every $n \in N$,
  we have $f_n^pg' = g_n$.  Since $g_n : T \to X_n$ by hypothesis, this means
  that $g' : T \to (f_n^p)^*X_n$ for every $n \in \N$.  Thus,
  $g' : T \to \bigcap_n (f_n^p)^*X_n = R$, and the mediating morphism can be
  lifted as expected.  After checking that this is natural in $T$, we conclude.
\end{proof}

\begin{remark}[A dual characterization]
  Since each inverse image function $f^* : \mathcal{R}_Y \to \mathcal{R}_X$
  preserves intersections and relations form complete lattices, we can build a
  corresponding left adjoint $f_! : \mathcal{R}_X \to \mathcal{R}_Y$, called the
  \emph{direct image by $f$}.  This allows us to find an alternative
  characterization of the bilimit $R$ above, by dualizing the proof.
  \[ R = \bigcup_n (f_n^e)_!X_n. \]
  Here, $\bigcup$ refers to the supremum of a family of relations, which exists
  by completeness.
\end{remark}

\begin{corollary}
  \label{cor:inverse-limit-lifting}
  If $\mathcal{C}$ satisfies the hypotheses of
  \Cref{thm:inverse-limit-construction}, then so does $\mathcal{R}$.
\end{corollary}

This leads to an alternative strategy for constructing recursive relations.

\begin{proof}[Proof of \Cref{thm:minimal-invariant-relations}; inverse limit construction]
  Thanks to \Cref{cor:inverse-limit-lifting}, we can apply
  \Cref{thm:inverse-limit-construction} to the lifting of $F$ in $\mathcal{R}$,
  and build a minimal invariant object $i_R : F(R,R) \cong R$ in $\mathcal{R}$.
  We can check that $pR$ is just $D$ up to isomorphism, since both are built as
  bilimits and those are preserved by $F$.  Thus, we might as well assume that
  $D = pR$ and $i_R = i$. The fact that $i$ is an isomorphism implies that
  $R \leq (i^{-1})^*F(R,R)$.  To conclude, we just need to show the reverse
  inequality.  We know that $i^{-1} : (i^{-1})^*F(R,R) \to F(R,R)$.  Since
  $i : F(R,R)\to R$, we find by composition that $1 : (i^{-1})^*F(R,R) \to R$,
  and we conclude that $R = (i^{-1})^*F(R,R)$.
\end{proof}

\section{Banach Fixed Point}
\label{sec:banach}

In the last two decades, \emph{guarded recursion} has emerged as a popular
method for defining recursive relations.  While originally developed for
reasoning about denotational semantics~\cite{Nakano00}, it was shortly after
adapted to the operational setting, where it proved to be a convenient interface
to \emph{step-indexed reasoning}~\cite{AppelM01,AppelMRV07}.

The basic idea is to work with a family of relations $(R_n)_{n \in \N}$.  In the
case of step indexing, $R_n$ represents a property that holds of terms of a
language within at most $n$ steps of computation, such as ``if the term
terminates in at most $n$ steps, then it is a value of type $\mathsf{bool}$''.
It is always possible to define such a family recursively if each $R_n$ depends
only on the values of $R_m$, for each $m < n$.  Manipulating such indices
directly quickly becomes cumbersome, so guarded recursion encapsulates this
process in a modality $\later$, usually known as ``later''.  Then, any recursive
definition becomes valid, as long as recursive occurrences of the relation
appear under $\later$.

After reviewing the basics of guarded recursion, we will see how it leads to an
alternative proof of \Cref{thm:minimal-invariant-relations}.  First, we need a
general, abstract setting where guarded definitions can be formulated.

\begin{definition}
  An \emph{ordered family of equivalences} (OFE) is a tuple
  $(X, (\stackrel{n}{=})_{n \in \N})$, where $X$ is a set,
  $(\stackrel{0}{=}) \supset (\stackrel{1}{=}) \supset \cdots$ is a decreasing
  sequence of equivalence relations on $X$, and $\stackrel{0}{=}$ is the total
  relation on $X$.  The family should converge to the identity on $X$:
  $\bigcap_n (\stackrel{n}{=}) = (=)$ or, equivalently,
  $(\forall n. x \stackrel{n}{=} y) \Rightarrow x = y$ for any $x, y \in X$.
\end{definition}

We will soon see examples of OFEs connected to the denotational models we have
been studying so far.  Before we get there, however, let us go back to the
example sketched above: an indexed family of relations on terms $t \in T$.  The
set of such indexed families forms an OFE: we say that $R \stackrel{n}{=} S$ if
and only if $R_m = S_m$ for any $m < n$.  Intuitively, this means that the two
relations are equivalent if we restrict ourselves to strictly less than $n$
steps of computation.  To define fixed points in such an abstract setting, we
require slightly more structure of OFEs.

\begin{definition}
  Let $X$ be an OFE. A \emph{Cauchy sequence} on $X$ is a sequence of elements
  $x : \N \to X$ such that, for every $n \in \N$, there exists $m \in \N$ such
  that, for any $i, j \geq m$, we have $x_i \stackrel{n}{=} x_j$.  We say that
  $X$ is a \emph{complete OFE} (COFE) if, for every Cauchy sequence $x$, there
  exists some (necessarily unique) $\lim x \in X$ such that, for every
  $n \in N$, there exists $m \in \N$ such that $x_i \stackrel{n}{=} \lim x$ for
  every $i \geq m$.
\end{definition}

We can show that the OFE of relations sketched above is complete.  Intuitively,
if we look at the $n$th level of the terms of a Cauchy sequence, they will
eventually stabilize at some $R_n$, and we can take the family of such $R_n$ to
be the limit of the sequence.

\begin{theorem}[Banach Fixed Point]
  \label{thm:banach}
  Suppose that a function $f : X \to X$ on a COFE is \emph{contractive}; that
  is, if $x \stackrel{n}{=} y$, then $f(x) \stackrel{n+1}{=} f(y)$.  Suppose,
  moreover, that there exists some $x_0 \in X$.  The sequence
  $x_i \triangleq f^i(x_0)$ is a Cauchy sequence, and its limit is the unique
  fixed point of $f$; $f(\lim x) = x$.
\end{theorem}

\begin{remark}[Connection to metric spaces]
  Each OFE $X$ gives rise to a metric space as follows: $d(x,y) = 2^{-n}$, where
  $n$ is the greatest $n$ such that $x \stackrel{n}{=} y$ holds (if there is no
  such $n$, then $x = y$, and we set $d(x,y) = 0$).  If $X$ is complete, then
  the resulting metric space is also complete.  If $f : X \to X$ is contractive,
  in the sense of \Cref{thm:banach}, then $d(f(x),f(y)) \leq \frac12 d(x,y)$,
  implying that $f$ is contractive in the traditional sense of metric spaces.
  This requirements guarantee that the usual metric formulation of Banach's
  fixed point theorem applies.
\end{remark}

\begin{remark}[Defining later]
  Given a family of relations $(R_n)_{n \in \N}$ as above, we can define another
  family $\later R$ by shifting $R$ by 1: $(\later R)_0 = T \times T$, and
  $(\later R)_{n+1} = R_n$.  To define a contractive function on families of
  relations, it suffices to consider functions of the form $f(R) = g(\later R)$,
  where $g$ is \emph{non-expansive}, which means that it preserves each relation
  $\stackrel{n}{=}$.  In this case, we say that the definition of $f$ is
  \emph{guarded}, which explains the connection to guarded recursion alluded to
  above.  Similar definitions of $\later$ can be stated for other types of
  relations.  Though it will not play a major role in what follows, guardedness
  is often a convenient way of checking that a definition is contractive
  (e.g. in a type theory).
\end{remark}

To apply \Cref{thm:banach} to construct recursive relations, we need to show
that relations on a minimal invariant $D$ form a COFE.  To this end, we restrict
ourselves to \emph{uniform relations}, which are the $R \in \mathcal{R}_D$ such
that $\pi_i : R \to R$ for every $i$.  As noted in \Cref{rem:uniformity}, the
relation that we aim to build is known to be uniform, so there is no harm in
restricting our search space to require uniformity from the start. We let
$\mathcal{U} \subset \mathcal{R}_D$ denote the set of uniform relations.

The reason for focusing on uniform relations is that they are entirely
determined by their inverse images by each of the $\pi_n$.  Indeed, if
$\pi_n^*R = \pi_n^*S$, then $R \leq \pi_n^*R = \pi_n^*S$, where the first
inequality holds by uniformity.  If this holds for every $n$, by admissibility,
$1 = (\lim_n \pi_n) : R \to S$ and $R \leq S$.  An analogous reasoning shows
that $S \leq R$, and we conclude that $R = S$.  This property, in turn, helps us
define a COFE structure over $\mathcal{U}$.

\begin{lemma}
  \label{lem:uniform-cofe}
  If $R$ and $S$ are uniform, then the following conditions are equivalent for
  every $n \in \N$:
  \begin{itemize}
  \item $\pi_n^*R = \pi_n^*S$
  \item $\pi_n : R \to S$ and $\pi_n : S \to R$.
  \end{itemize}

  If one of these conditions holds, we say that $R \stackrel{n}{=} S$.  This
  assignment endows the set $\mathcal{U}$ with the structure of a COFE.
\end{lemma}

\begin{proof}
  The equivalence between the two notions follows from the previous discussion.
  To show that this indeed defines a COFE, note that we have already seen that
  $R = S$ when $R \stackrel{n}{=} S$ for every $n$, so we have a well defined
  OFE.  Thus, we just need to prove completeness.  Let $(R_i)$ be a Cauchy
  sequence on $\mathcal{U}$.  For every $i \in \N$, there exists some
  $m_i \in \N$ such that $\pi_i^*(R_n)$ is equal to
  $S_i \triangleq \pi_i^*(R_{m_i})$ for any $n \geq m_i$.  Without loss of
  generality, we can assume that $(m_i)$ is increasing.  We pose
  \begin{align*}
    \lim R & \triangleq \bigcap_i S_i.
  \end{align*}
  We can show that uniform relations are closed under inverse images by $\pi_i$
  and intersections, hence each $S_i$ and $\lim R$ are indeed uniform.
  Moreover, because $S$ is a subsequence of $R$, it must be a Cauchy sequence,
  and it must have the same limit as $R$, if one of them does have a limit.

  To conclude, we just need to show that $\lim R$ is indeed the limit of $S$.
  First, note that $(S_i)$ is decreasing.  Indeed, given $i \leq j$, we have
  \begin{align*}
    S_i & = \pi_i^*(R_{m_i}) = \pi_i^*(R_{m_j}) = (\pi_j \circ \pi_i)^*(R_{m_j})
     = \pi_i^*(\pi_j^*(R_{m_j})) = \pi_i^*(S_j).
  \end{align*}
  Thus, $S_j \leq S_i = \pi_i^*(S_j)$ is equivalent to $\pi_i : S_j \to S_j$,
  which follows from the uniformity of $S_j$.

  On the other hand, given $i \in \N$, we have $\lim R \stackrel{i}{=} S_i$.
  Indeed,
  \begin{align*}
    \pi_i^*(\lim R) & = \pi_i^*\left(\bigcap_j S_j\right) \\
    & = \pi_i^*\left(\bigcap_{j \geq i} S_j\right)
    & \text{$S$ is decreasing} \\
    & = \bigcap_{j \geq i} \pi_i^*(S_j)
    & \text{intersections commute with inverse images} \\
    & = \bigcap_{j \geq i} \pi_i^*(S_i) \\
    & = \pi_i^*(S_i),
  \end{align*}
  which shows that $S$ does converge to $\lim R$.
\end{proof}

Now that we have a COFE, we just need a contractive operator on
$\mathcal{U}$.

\begin{lemma}
  \label{lem:contractive-action}
  The following defines a contractive operator on $\mathcal{U}$:
  \begin{align*}
    \Psi(R) & \triangleq (i^{-1})^*(F(R,R)).
  \end{align*}
\end{lemma}

\begin{proof}
  We begin with the following auxiliary result.  If $R \stackrel{n}{=} S$, for
  $R, S \in \mathcal{U}$, then
  \[ \pi_{n+1} = iF(\pi_n,\pi_n)i^{-1} : \Psi(R) \to \Psi(S). \] Indeed, by
  unfolding definitions, we have $i^{-1} : (i^{-1})^*F(R,R) \to F(R,R)$ and
  $i : F(S,S) \to (i^{-1})^*F(S,S)$.  By unfolding $\Psi$, and by composition,
  we can prove this statement by showing
  \[ F(\pi_m,\pi_m) : F(R,R) \to F(S,S). \] This follows from
  $R \stackrel{n}{=} S$ by \Cref{lem:uniform-cofe}.

  Let us proceed with the main proof. First, note that $\Psi(R)$ is indeed
  uniform, so $\Psi : \mathcal{U} \to \mathcal{U}$.  Indeed, we need to show
  that $\pi_n : \Psi(R) \to \Psi(R)$ for any $n$.  If $n \neq 0$, we apply the
  auxiliary result above.  If $n = 0$, it suffices to show that
  $i^{-1}\bot : \Psi(R) \to F(R,R)$.  But
  $i^{-1}\bot = i^{-1}\bot\bot \leq i^{-1}i\bot = \bot$, so $i^{-1}\bot = \bot$,
  and we conclude because $\mathcal{R} \ni F(R,R)$ is pointed.  Second, we need
  to show that $R \stackrel{n}{=} S$ implies
  $\Psi(R) \stackrel{n+1}{=} \Psi(S)$.  This follows by applying the auxiliary
  result in both directions, and by using \Cref{lem:uniform-cofe}.
\end{proof}

Combining all these ingredients, we obtain yet another strategy for building
$R_D$.

\begin{proof}[Proof of \Cref{thm:minimal-invariant-relations}; Banach fixed
  point]

  It suffices to apply \Cref{thm:banach} to the operator
  $\Psi : \mathcal{U} \to \mathcal{U}$ of \Cref{lem:contractive-action}. We just
  need to find an initial uniform relation to construct the fixed point.  Note
  that $\mathcal{R}_D$ has an element $\top$, defined as the intersection of the
  empty family of relations.  Moreover, for any $f : X \to D$ in $\mathcal{C}$
  and $R \in \mathcal{R}_X$, we have $f : R \to \top$ by
  \Cref{lem:relational-structure-properties}.  In particular,
  $\pi_i : \top \to \top$ for any $i$, so $\top \in \mathcal{U}$ and we
  conclude.
\end{proof}

\section{Kleene Fixed Point}
\label{sec:kleene}

As observed earlier, the inverse limit construction can be seen as a
generalization of Kleene's fixed point theorem:
\begin{theorem}[Kleene]
  \label{thm:kleene}
  Let $X$ be a pointed CPO and $f : X \to X$ be a continuous function.  Then $f$
  has a least fixed point $x = f(x)$, given by the limit of the chain
  $\bot \sqsubseteq f(\bot) \sqsubseteq f^2(\bot) \sqsubseteq \cdots$.
\end{theorem}
As a minor variation on \Cref{sec:inverse-limit}, let us sketch how we can
restate those results using Kleene's theorem, by viewing domains and relations
as an ordered structure rather than a category.  There are two issues that we
need to address.  First, $\mathcal{R}$ contains potentially many morphisms
between a pair of objects, whereas a CPO $X$ seen as a category has at most one.
Second, $\mathcal{R}$ is not a \emph{skeletal category}: there are objects that
are isomorphic, but not equal.  By contrast, a CPO seen as a category is
skeletal because its order is antisymmetric.

To solve the first issue, consider the slice category $\mathcal{C}^e/D$.
Objects of $\mathcal{C}^e/D$ are embeddings of type $X \to D$, and arrows from
$X \to D$ to $Y \to D$ commuting triangles of embeddings:
\begin{center}
  \begin{tikzcd}
    X \ar{dr} \ar{rr} & & Y \ar{dl} \\
    & D. &
  \end{tikzcd}
\end{center}
Since embeddings are monomorphisms, if there are two arrows of type $X \to Y$ in
$\mathcal{C}^e/D$, they must be equal.  We can apply a similar idea to
$\mathcal{R}$ by considering $\mathcal{R}(\mathcal{C}^e/D)$, which is defined as
the following pullback:
\begin{center}
  \begin{tikzcd}
    \mathcal{R}(\mathcal{C}^e/D) \ar{r} \ar{d}
    \ar[phantom, very near start]{rd}{\lrcorner} &
    \mathcal{R}^e \ar{d} \\
    \mathcal{C}^e/D \ar{r} & \mathcal{C}^e
  \end{tikzcd}
\end{center}
Explicitly, objects of $\mathcal{R}(\mathcal{C}^e/D)$ are triples
$X = (|X|, e_X : |X| \to D,R_X : \mathcal{R}_X)$, where $e_X$ is an embedding.
An arrow $f : X \to Y$ is an embedding $f^e : |X| \to |Y|$ such that
$e_Yf^e = e_X$ and such that, in $\mathcal{R}$, we have $f^e : R_X \to R_Y$ and
$f^p : R_Y \to R_X$.  Once again, there is at most one arrow of any given type
in this category.

To solve the second issue, note that, in many cases of interest, we can replace
$\mathcal{C}^e/D$ (and $\mathcal{R}(\mathcal{C}^e/D)$) with equivalent skeletal
subcategories, by choosing canonical representatives for their objects.  For
instance, if $\mathcal{C}$ is $\CPO_\bot$, we can replace an embedding
$e_X : X \to D$ with its image in $D$, which is isomorphic to $X$. Two objects
in $\mathcal{C}^e/D$ are isomorphic if and only if their images in $D$ are
equal.  In what follows, I'll assume that such canonical representatives exist,
and that $D \to D$ is its own representative.  By abuse of notation, I'll
identify the above categories with their skeletal equivalents.

Both $\mathcal{C}^e/D$ and $\mathcal{R}(\mathcal{C}^e/D)$ are CPOs: to compute
the least upper bound of a chain, we simply project the chain onto $\mathcal{C}$
(or $\mathcal{R}$), compute its bilimit, and use its canonical representative in
$\mathcal{C}^e/D$ (or $\mathcal{R}(\mathcal{C}^e/D)$).  Moreover, these CPOs are
pointed: their least elements are $(1,\bot : 1 \to D)$ and
$(1,\bot : 1 \to D, \top)$. This leads to the following alternative proof.

\begin{proof}[Proof of \Cref{thm:minimal-invariant-relations}; Kleene fixed
  point]

  Let $F : \mathcal{C}^{op} \times \mathcal{C} \to \mathcal{C}$ be a
  $\CPO$-functor.  As we have seen in \Cref{sec:inverse-limit}, we can view the
  admissible action of $F$ on $\mathcal{R}$ as a lifting
  $F_{\mathcal{R}} : \mathcal{R}^{op} \times \mathcal{R} \to \mathcal{R}$.
  These functors give rise to functors $F^e : \mathcal{C}^e \to \mathcal{C}^e$
  and $F^e_{\mathcal{R}} : \mathcal{R}^e \to \mathcal{R}^e$ that preserve
  colimits of chains.  We have the following commutative diagram:
  \begin{center}
    \begin{tikzcd}
      \mathcal{R}^e \ar{r}{F^e_{\mathcal{R}}} \ar{d} & \mathcal{R}^e \ar{d} \\
      \mathcal{C}^e \ar{r}{F^e} & \mathcal{C}^e.
    \end{tikzcd}
  \end{center}

  By working with canonical representatives, we can view these functors as
  continuous functions $f : \mathcal{C}^e/D \to \mathcal{C}^e/D$ and
  $f_{\mathcal{R}} : \mathcal{R}(\mathcal{C}^e/D) \to
  \mathcal{R}(\mathcal{C}^e/D)$, and we can take their fixed points by
  \Cref{thm:kleene}.  By construction, the fixed point of $f$ is just $D$, and
  the above diagram implies that
  $p(\mathsf{fix}(f_{\mathcal{R}})) = \mathsf{fix}(f) = D$, where
  $p : \mathcal{R}(\mathcal{C}^e/D) \to \mathcal{C}^e/D$ is the canonical
  projection.  This means, after some unfolding, that the relation component of
  $\mathsf{fix}(f_{\mathcal{R}})$ is a relation on $D$ that satisfies the
  recursive equation we are seeking.
\end{proof}

\begin{remark}
  Most categories used in domain theory have canonical representatives of
  embeddings---we can take the image of an embedding, as we have done above, or
  we can choose representatives using the axiom of choice.  But if images are
  not available, there is another option that does not rely on the axiom of
  choice: to work with $\bar{\mathcal{C}}$, the \emph{Karoubi envelope} of
  $\mathcal{C}$.  This category extends $\mathcal{C}$ by freely splitting all
  idempotent arrows in $\mathcal{C}$ (that is, arrows $p : X \to X$ such that
  $pp = p$).  Roughly, this means that $\bar{\mathcal{C}}$ contains canonical
  image objects of all idempotents in $\mathcal{C}$.  In particular, we can
  compute the image of the idempotent $f^ef^p$ determined by an embedding
  $f^e : Y \to X$, which yields a choice of representatives for embeddings.
  Moreover, $\bar{\mathcal{C}}$ (and $\bar{\mathcal{R}}$) inherit the properties
  of the original categories that we relied on to carry the above constructions,
  so our results still apply.
\end{remark}

\section{Conclusion}
\label{sec:conclusion}

We have just reviewed Pitts' framework of relational properties of
domains~\cite{Pitts96} and seen how it relates to other important fixed-point
theorems: the \emph{inverse limit construction}~\cite{SmythP82}, \emph{Banach's
  fixed-point theorem}, and \emph{Kleene's fixed point theorem}.  These
connections are implicit in some of the existing literature, and probably
already known by experts.  For example, the work of \textcite{HermidaJ98}
presents a different method for constructing relations on recursive data types
that requires lifting limits and colimits along a fibration; likewise, the proof
of Pitts' method with the inverse limit construction uses
\Cref{lem:bilimits-relational-structures}, which lifts bilimits to a relational
structure.  As for Banach's fixed-point theorem, several works for reasoning
about denotational models~\cite{DBLP:conf/tldi/BirkedalST09, BirkedalST09,
  BirkedalRSSTY11} employ similar constructions while sometimes noting that
Pitts' framework could have been used
instead~\cite{DBLP:conf/tldi/BirkedalST09}.  Here, we have seen how this
connection goes beyond the construction of a particular set of logical
relations, and lies at the heart of Pitts' method.  It is worth noting that the
connections between these fixed-point theorems go beyond the setting of
relational reasoning---e.g. \textcite{Thamsborg10} discusses how we can view
Banach's fixed-point theorem as an instance of Kleene's.

Traditionally, \emph{step-indexing} uses the steps in some operational semantics
to define recursive relations~\cite{AppelM01}.  In light of the connections
explained above, Pitts' construction---as well as other applications of Banach's
fixed-point theorem for denotational models~\cite{BirkedalRSSTY11, BirkedalST09,
  DBLP:conf/tldi/BirkedalST09, MacQueenPS86}---use a similar trick to ensure
that the recursion is well-founded, but count the \emph{number of unfoldings of
  a recursive type} instead.  In this sense, guarded recursion is more general,
since the notion of counting can be tied to anything that can be tracked in the
execution of a program, not just the number of unfoldings of the domain
equation.  On the other hand, relations constructed with Pitts' method are often
cleaner then their guarded counterparts, because they do not have to mention
step indices or guards explicitly.

One question that I have not explored is how this connects to variants of Pitts'
construction used for operational semantics, as developed by
\textcite{BirkedalH99} or \textcite{CraryH07}.  Such works note that the
projections $\pi_i$ can often be defined as regular programs in a language, and
leverage this fact to adapt Pitts' ideas to establish powerful reasoning
principles for program equivalence. Like Pitts' original construction, these
works employ the Knaster-Tarski fixed-point theorem, but I believe that it might
be possible to adapt their constructions to leverage other results as well. One
possible connection lies in the proof of metric preservation for the Fuzz
language~\cite{ReedP10}.  Its argument employed step-indexed logical relations,
but the indices of the relations tracked the number of recursive unfoldings
reduced during execution rather than the number of transitions in a small-step
semantics.  This idea is similar to constructions by guarded recursion performed
in denotational settings~\cite{BirkedalST09, DBLP:conf/tldi/BirkedalST09,
  BirkedalRSSTY11, MacQueenPS86}, suggesting that it might be possible to obtain
an alternative, operational proof of metric preservation for Fuzz along the
lines of \textcite{BirkedalH99,CraryH07}.

\section*{Acknowledgments}

I would like to thank Lars Birkedal for useful discussions on this topic.

\printbibliography

\end{document}